\documentclass[letterpaper]{article} 
\usepackage{aaai25}  
\usepackage{times}  
\usepackage{helvet}  
\usepackage{courier}  
\usepackage[hyphens]{url}  
\usepackage{graphicx} 
\usepackage{natbib}  
\usepackage{caption} 
\usepackage{algorithm}
\usepackage[noend]{algpseudocode} 

\usepackage{subcaption}  
\usepackage{hhline}
\usepackage{multirow}
\usepackage{placeins}
\usepackage{amssymb}

\usepackage{amsthm}

\newtheorem{theorem}{Theorem}
\newtheorem{lemma}{Lemma}

\newtheorem{definition}{Definition}
\newcommand{\comp}{\mathcal{C}}
\newcommand{\framework}{WinC-MAPF}

\usepackage[table,xcdraw]{xcolor}
\usepackage[modulo,switch]{lineno} 

\usepackage{amsmath}

\DeclareMathOperator*{\argmin}{arg\,min}

\usepackage{newfloat}
\usepackage{listings}
\DeclareCaptionStyle{ruled}{labelfont=normalfont,labelsep=colon,strut=off} 
\floatstyle{ruled}
\newfloat{listing}{tb}{lst}{}
\floatname{listing}{Listing}
\title{Windowed MAPF with Completeness Guarantees}
\author{
    Rishi Veerapaneni, Muhammad Suhail Saleem, Jiaoyang Li, Maxim Likhachev
}
\affiliations{
    Carnegie Mellon University \\
    \{rveerapa, msaleem2, jiaoyanl, mlikhach\}@cs.cmu.edu
}

\begin{document}

\maketitle

\begin{abstract}
Traditional multi-agent path finding (MAPF) methods try to compute entire collision free start-goal paths, with several algorithms offering completeness guarantees. However, computing partial paths offers significant advantages including faster planning, adaptability to changes, and enabling decentralized planning. 
Methods that compute partial paths employ a ``windowed" approach and only try to find collision free paths for a limited timestep horizon. While this improves flexibility, this adaptation introduces incompleteness; all existing windowed approaches can become stuck in deadlock or livelock. 
Our main contribution is to introduce our framework, \framework, for Windowed MAPF that enables completeness. Our framework leverages heuristic update insights from single-agent real-time heuristic search algorithms and agent independence ideas from MAPF algorithms. We also develop Single-Step Conflict Based Search (SS-CBS), an instantiation of this framework using a novel modification to CBS.
We show how SS-CBS, which only plans a single step and updates heuristics, can effectively solve tough scenarios where existing windowed approaches fail.
\end{abstract}


\section{Introduction}
A core problem for multi-agent systems is to figure out how agents should move from their current locations to their goal locations. Without careful consideration, agents can collide, get stuck in deadlock, or take inefficient paths which take longer to traverse. This Multi-Agent Path Finding (MAPF) problem is particularly tough in congestion or when the number of agents becomes very large (e.g. 100s).

Full horizon MAPF methods attempt to find entire paths for each agent from their start to their goal. Many such MAPF methods additionally have theoretical completeness guarantees, i.e., they will find a solution if one exists given enough time and compute. 
However, planning partial paths has multiple advantages including decreasing planning time, adaptability to changes, and enabling decentralized planning. This is particularly useful in scenarios where planning a full horizon path may be tough (e.g. too many collisions to resolve) and planning a partial path is more feasible.


\begin{table}[t!]
\resizebox{0.99\linewidth}{!}{
\begin{tabular}{c|c||cc|cc|cc}
\multicolumn{1}{c}{} & \multicolumn{1}{c}{} & \multicolumn{2}{c}{\includegraphics[width=0.2\linewidth]{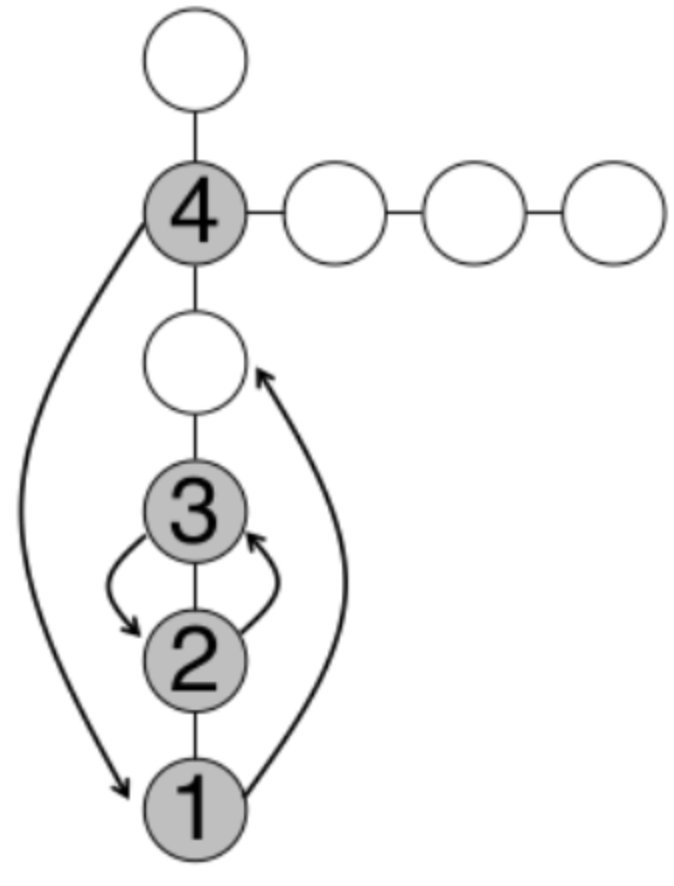}} &  \multicolumn{2}{c}{\includegraphics[width=0.2\linewidth]{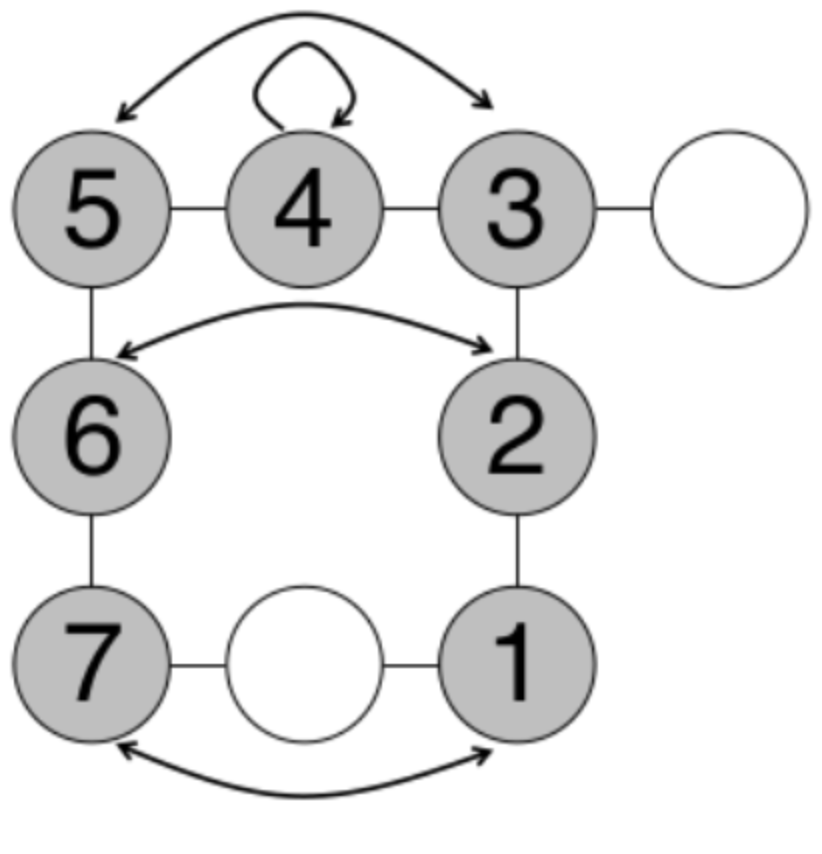}} &  \multicolumn{2}{c}{\includegraphics[width=0.2\linewidth]{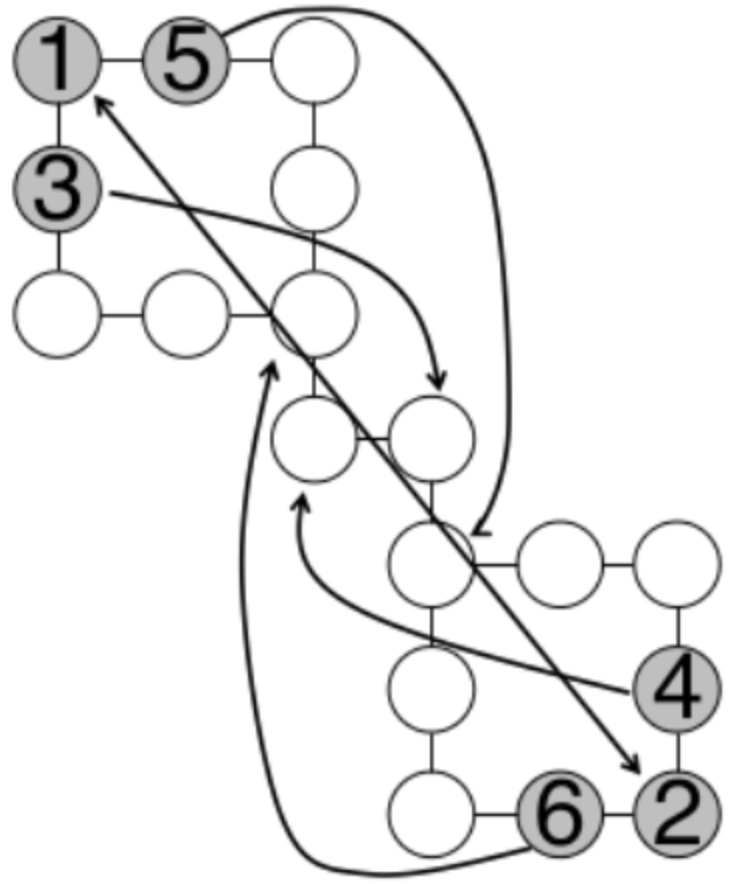}} \\
 &  & \multicolumn{2}{c|}{Tunnel} & \multicolumn{2}{c|}{Loopchain} & \multicolumn{2}{c}{Connector} \\
\multirow{-2}{*}{Method} & \multirow{-2}{*}{Horizon} & 3 & 4 & 6 & 7 & 5 & 6 \\ \hline
wCBS & 1,2,4,8,16 & \cellcolor[HTML]{FFCCC9}- & \cellcolor[HTML]{FFCCC9}- & \cellcolor[HTML]{FFCCC9}- & \cellcolor[HTML]{FFCCC9}- & \cellcolor[HTML]{FFCCC9}- & \cellcolor[HTML]{FFCCC9}-\\
CBS+ & $\infty$ & \cellcolor[HTML]{9AFF99}0.9 & \cellcolor[HTML]{FFCCC9}- & \cellcolor[HTML]{FFCCC9}- & \cellcolor[HTML]{FFCCC9}- & \cellcolor[HTML]{FFCCC9}- & \cellcolor[HTML]{FFCCC9}- \\ \hline
SS-CBS & 1 & \cellcolor[HTML]{9AFF99}1 & \cellcolor[HTML]{9AFF99}1 & \cellcolor[HTML]{9AFF99}1 & \cellcolor[HTML]{9AFF99}0.95 & \cellcolor[HTML]{9AFF99}1 & \cellcolor[HTML]{9AFF99}1 \\
\end{tabular}}
\caption{Success rates of different optimal approaches: windowed CBS (wCBS) and full horizon CBS with all optimizations enabled (CBS+). Horizon denotes window size. We test on small congested scenarios (number of agents written below each map) across 20 seeds and a 1 minute timeout. wCBS and CBS+ fail due to deadlock or timeout while
SS-CBS, with a \textit{single-step planning horizon}, is able to solve these scenarios using our \framework{} framework which maintains completeness guarantees.}
\label{tab:sneak-peak}
\vspace{-1em}
\end{table}

Therefore existing windowed methods define a time window $W$ and plan paths for each agent to the goal such that the first $W$ timesteps account for inter-agent coordination and avoids collisions. This window $W$ is typically much shorter than the entire solution path, e.g. $W=5$ is common when the entire solution path spans 50 to 500 timesteps. As a result, windowed methods are significantly faster than those that compute the entire path.

A key issue with windowed approaches is that their myopic planning results in deadlock or livelock. 
Table \ref{tab:sneak-peak} shows examples where windowed MAPF solvers fail in congestion if their window is too small.
More broadly, all existing windowed MAPF solvers regardless of window size lack theoretical completeness, and several windowed works have explicitly cited deadlock as a key issue in their experiments \cite{cooperativeSilver2005,rhcrLi2020,pibt,jiang2024scaling_mapf_competition}. Interestingly, no previous work has addressed completeness in windowed MAPF solvers.

Our first main contribution is the introduction of the \underline{Win}dowed \underline{C}omplete MAPF framework, \framework, designed to create Windowed MAPF solvers that guarantee completeness. 
\framework{} is a general framework that leverages concepts from single-agent real-time heuristic search and the semi-independent structure of MAPF. First, we view Windowed MAPF in its joint configuration and show how we can apply real-time heuristic updates on the joint configuration to enable completeness. However, due to the large joint configuration space in a MAPF problem, naively applying heuristic updates is impractical. Thus, second, we leverage the semi-independent structure of MAPF problems to focus heuristic updates on disjoint agent groups simultaneously, resulting in efficient performance. 
An important module in the \framework{} framework is an Action Generator (AG) that computes the next set of actions the agents need to execute. To guarantee completeness, the AG must (1) find the optimal windowed action that incorporates heuristic updates, and (2)  identify disjoint agent groups. Our framework and proof encourage future research on the development of windowed AGs that satisfy these properties. 

As a starting point, our second main contribution is developing Single-Step Conflict-Based Search (SS-CBS), a CBS-based AG that follows the \framework{} framework and plans only for a single timestep ($W=1$). Conflict-Based Search \cite{sharon2015cbs} can be easily modified to be windowed by only considering conflicts within the window, but we show how naively integrating heuristic updates can fail. Thus, SS-CBS introduces a novel ``heuristic conflict" and constraint to address this issue. We empirically demonstrate how SS-CBS, with \textit{single} step planning, outperforms windowed CBS with larger windows across both small and large instances.

\section{Related Work}

\subsection{Problem Formulation}
Multi-Agent Path Finding (MAPF) is the problem of finding collision-free paths for a group of $N$ agents ${i = 1, \dots, N}$, that takes each agent from its start location $s_i^{\text{start}}$ to its goal location $s_i^{\text{goal}}$. In traditional 2D MAPF, the environment is discretized into grid cells, and time is broken down into discrete timesteps. Agents are allowed to move in any cardinal direction or wait in the same cell. A valid solution is a set of agent paths $\Pi = \{ \pi_1, ..., \pi_N \}$ where $\pi_i[0] = s_i^{start}$, $\pi_i[T_i] = s_i^{goal}$ where $T_i$ is the maximum timestep of the path for agent $i$.
Critically, agents must avoid vertex collisions (when $\pi_i[t] = \pi_{j \neq i}[t]$) and edge collisions (when $\pi_i[t] = \pi_j[t+1] \wedge \pi_i[t+1]=\pi_j[t]$) for all timestep $t$. 
The typical objective in optimal MAPF is to find a solution $\Pi$ that minimizes $|\Pi| = \sum_{i=1}^N |\pi_i| = \sum_{i=1}^N \sum_{t=0}^{T_i-1} c(s_i^t,s_i^{t+1})$. 
This work solves standard MAPF which has $c(s_i^t,s_i^{t+1})=1$ unless the agent is resting at its goal (where $c(s_i^{\text{goal}},s_i^{\text{goal}})=0$).

\subsubsection{Windowed MAPF}
Instead of resolving all collisions, windowed planners iteratively plan a smaller collision-free path and execute, in essence breaking the problem into smaller more feasible chunks.
For example, a recent MAPF competition, League of Robot Runners \cite{chan2024_league_robot_runners_competition}, required planning for hundreds of agents within 1 second and utilized windowed planning interleaved with execution. 

Windowed MAPF methods plan partial paths that only reason about collisions for timesteps $t \leq W$ where $W$ is a hyper-parameter window size. Thus after the first $W$ timesteps, the remaining path $\pi_i[W] ... \pi_i[T_i]$ is the agent's optimal path to the goal (in the absence of other agents) as it does not need to avoid collisions with other agents.
Mathematically then, the cost of $\pi_i$ is $\sum_{t=0}^{W-1} c(s_i^t,s_i^{t+1}) + c^*(s_i^W,s_i^{goal})$ where $c^*(s_i^W,s_i^{goal})$ is the optimal cost to the goal.
We note that all performant 2D MAPF methods compute a backward dijkstra's for each agent where $h^*_i(s) = c^*(s,s_i^{goal})$. Thus instead of fully planning $\pi_i^{0:T_i}$, we can equivalently plan just the windowed horizon $\pi_i^{0:W}$ and minimize $|\Pi| = \sum_{i=1}^N ( \sum_{t=0}^{W-1} c(s_i^t,s_i^{t+1}) + h^*_i(s_i^W))$.

\subsection{MAPF Methods}
There exist many different types of heuristic search solvers for MAPF. 
One old approach is Prioritized Planning \cite{erdmann1987multiple} which assigns priorities to agents and plans them sequentially with later agents avoiding earlier agents. PIBT \cite{pibt} is a recent popular method that allows agents to ``inherit" other agents' priorities. Conflict Based Search \cite{sharon2015cbs} is another popular method that decoupled the planning problem into two stages. A high-level search resolves collisions between agents by applying constraints while a low-level search finds paths for individual agents that satisfy constraints. There are many extensions to CBS that improve the searches as well as the applied constraints \cite{barer2014suboptimalecbs,li2021eecbs,srli2021}.

When faced with shorter planning times, methods typically simplify the planning problem to just find partial collision-free paths. Windowed Hierarchical Cooperative A* \cite{cooperativeSilver2005} is a windowed variant of Hierarchical Cooperative A* which is essentially a prioritized planner using a backward Dijkstra's heuristic, and is not complete due to their use of priorities. Rolling Horizon Conflict Resolution (RHCR) applies a rolling horizon for lifelong MAPF planning and replans paths at repeated intervals \cite{rhcrLi2020}. RHCR faces deadlock and attempts to combat it by increasing the planning window but still notes that their method is incomplete. Bounded Multi-Agent A* \cite{mapf-real-time-bmaa-2018} proposes that each agent runs its own limited horizon real-time planner considering other agents as dynamic obstacles. However, the method acknowledges that deadlock occurs when agents need to coordinate with one another.

Planning and Improving while Executing \cite{zhang2024pie} is a recent work that attempts to quickly generate an initial full plan 
using LaCAM \cite{okumura2022lacam} and then refines it during execution using LNS2 \cite{li2022mapf-lns2}. 
However, if a complete plan cannot be found, the method resorts to using the best partial path available, making it incomplete in such situations. The winning submission \cite{jiang2024scaling_mapf_competition} to the Robot Runners competition, due to the tight planning time constraint, leveraged windowed planning of PIBT with LNS \cite{li2021mapf-lns}. They explicitly note deadlock in congestion is a significant challenge.

To the best of our knowledge, there does not exist any windowed MAPF solver with completeness guarantees.

\subsection{Real-Time Single Agent Search} \label{sec:real-time-search-backgroun}
We leverage ideas from ``Real-Time" Search, a single-agent heuristic search problem where, due to limited time constraints, the agent is required to iteratively plan and execute partial paths. Despite repeatedly planning partial paths, Real-Time Search methods are designed to maintain completeness. The main innovation in single-agent real-time search literature is that the agent updates (increases) the heuristic value of encountered states. This prevents deadlock/livelock as states that are repeatedly visited have larger and larger heuristic values which encourages the search to explore other areas. A large variety of real-time algorithms such as LRTA* \cite{korf1990_lrta}, RTAA* \cite{koenig2006_rtaa}, and LSS-LRTA* \cite{koenig2009_lss_lrta} propose to update the heuristic in different ways. 

Section \ref{sec:lrta-proof} contains a formal proof of how completeness guarantees hold in Real-Time Search. 
The core idea is that given a current state $s$, cost function $c$, and a partial path leading to a new state $s^W$, we update the heuristic value via a standard bellman update equation $h(s) \gets U(s,s^W):=\max(h(s), c(s,s^W) + h(s^W))$.
Now if an agent does not reach the goal, it must be stuck in deadlock/livelock and repeatedly visiting some states $s \in \mathcal{S}_{stuck}$. Under certain (achievable) conditions, updating the heuristic can be shown to repeatedly increase $h(s)$ for $s \in \mathcal{S}_{stuck}$ \cite{korf1990_lrta}.  
At some point then, if the agent is optimally planning to the state $s^W$ that minimizes $c(s,s^W) + h(s^W)$, it will eventually pick an $s^W \notin \mathcal{S}_{stuck}$. Thus most single-agent Real-Time Search algorithms require an optimal planner to maintain completeness.
 
We note that in the above proof sketch, the planning window $W$ does not matter, nor does how many steps in the partial plan the agent chooses to move (e.g. committing to only one step or executing all $W$ steps). Completeness holds regardless if $W \geq 1$ and the agent moves at least one step.


\section{Windowed-MAPF with Guarantees} 
This section describes our \underline{Win}dowed \underline{C}omplete MAPF framework, \framework, for creating windowed MAPF solvers that guarantee completeness. We leverage two key insights.
Our first insight is that we can apply single-agent real-time update ideas to MAPF planning if we interpret the MAPF problem as a single-agent problem in the combined joint configuration space. This allows us to update the heuristics of previously seen states enabling completeness. However, just doing this is ineffective due to the large state space. To this extent, our second insight is that we can leverage MAPF's agent semi-independence to intelligently update the heuristic value of multiple states, allowing the search to fill in heuristic depressions quickly and exit local minima faster.

\begin{figure*}[t!]
    \centering
    \includegraphics[width=0.7\textwidth]{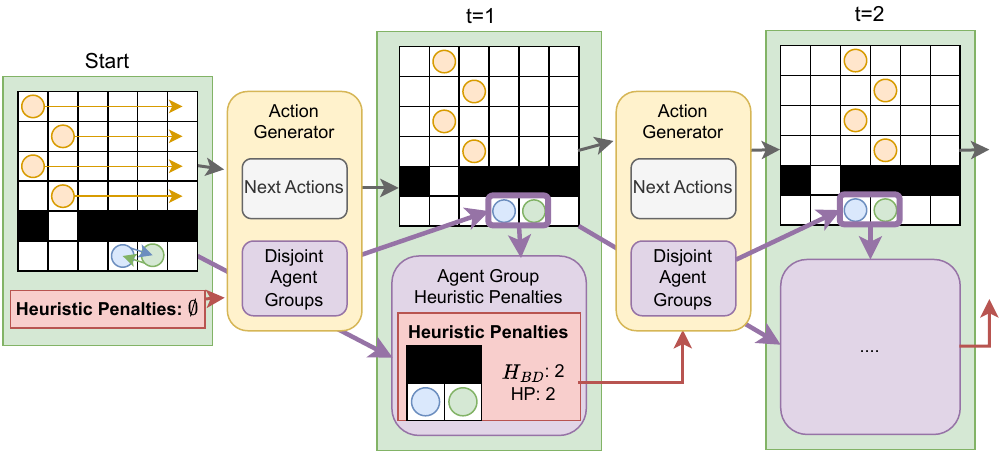}
    \caption{The iterative planning, execution, and grouped heuristic updates/penalties for \framework. }
    \label{fig:action-generator-framework}
    \vspace*{-1em}
\end{figure*}

\subsection{Planning in Joint Configuration Space} \label{sec:planning-configuration-space}
We can leverage single-agent Real-Time Search literature (Section \ref{sec:real-time-search-backgroun}) if we view our multi-agent problem in the joint configuration space of all agents. 

Here is a quick conceptual but mathematically imprecise summary: We view the $N$ agents in their joint configuration $\comp$. At every timestep, we query the action generator to return a valid sequence of actions that leads to a configuration $\comp'$ which minimizes $c(\comp, \comp') + h(\comp')$. We move the agent along the path as well as update the heuristic of $\comp$ via $h(\comp) \gets c(\comp, \comp') + h(\comp')$. We can then prove completeness by directly applying single agent real-time search literature.

We now formally define the above.
Given $N$ agents, we define the configuration $\comp^t = [s_1^t, ... s_N^t]$. At every planning iteration, we query a high-level ``action generator" to return a sequence of configurations $\Pi^{0:W} = [\comp^0, ..., \comp^W]$ which minimizes $|\Pi| = |\Pi^{0:W}| + |\Pi^{W:T_{max}}| = \sum_{t=0}^{W-1} c(\comp^t,\comp^{t+1}) + h_{BD}(\comp^W)$.
We define the joint cost and heuristic intuitively, $c(\comp, \comp') = \sum_{i=1}^N c(s_i,s_i')$ and $h_{BD}(\comp) = \sum_{i=1}^N h^*_i(s_i)$ (BD = Backward Dijkstra).

\textbf{Heuristic Penalties}
Now given the transition (partial path) $\comp \rightarrow \comp^W$ from the AG, we update 
\begin{equation} \label{eq:update}
    h(\comp) \gets U(\comp, \comp^W) := 
    \max(h(\comp), c(\comp, \comp^W) + h(\comp^W))
\end{equation}

$h(\comp)$ is initially set to $h_{BD}(\comp)$ and gradually increases as agents visit configurations. We use the term ``heuristic penalty" (HP) configuration for visited configurations whose heuristic increased via the update equation to denote how the updated heuristic ``penalizes" those configurations and encourages the search to explore others. An HP configuration $\comp$ has a nonzero increase from the base heuristic value, i.e. $h(\comp) = h_{BD}(\comp) + h_{p}(\comp)$ where the penalty $h_p(\comp) > 0$. 

\textbf{Action Generator}
Given $h(\comp)$ which includes states with HPs and without, we want an optimal action generator AG that finds $\argmin_{C{^W}} c(\comp, \comp^W) + h(\comp^W)$. If so, we get the same completeness guarantees by directly reusing the single-agent proof. 
We require an optimal windowed AG for our proof of completeness but conjecture that future work could relax this.

\subsection{Disjoint Agent Groups} \label{sec:coupled-agents}
The framework we have described so far suffers from an obvious issue; the number of configurations $\comp$ grows exponentially with the number of agents. 
Consequently, escaping local minima, which in our context are agents stuck in deadlock/livelock, can require updating the heuristic of an impractical number of states. 


Our key idea is therefore to compute and apply heuristic penalties to specific groups of agents instead of on the entire configuration. This significantly speeds up performance as it allows the search to focus on specific ``stuck" agents instead of on all agents. 
Thus given a configuration transition $\comp \rightarrow \comp^W$, we decompose all the agents into disjoint agent groups such that agents between different groups are not interacting with each other. For example, in Figure \ref{fig:action-generator-framework}, we would like to determine that the blue and green agents are blocking each other but that the orange agents are not.
\begin{definition}[Disjoint Agent Groups]
    Given a configuration transition $\comp \rightarrow \comp^W$, and set of disjoint agent groups $\{Gr^i\}$, we have the property that for each agent $R^j$ with transition $s_j \rightarrow s_j^W$ in disjoint agent group $Gr^i$, there cannot exist another agent in a different group $Gr^k$ that blocked $R^j$ from picking a better path.
\end{definition}
Conceptually, this means that each group's decision to make $\comp_{Gr^i} \rightarrow \comp_{Gr^i}^W$ is independent of agents in other groups. Instead of computing disjoint groups where agents are not interacting, it is easier to determine coupled agent groups where they could be interacting.

\begin{definition}[Coupled Agents]
    Given a configuration transition $\comp \rightarrow \comp^W$, an agent $R_i$ is coupled with $R_j$ if $R_j$ prevents $R_i$ from choosing a better path or vice-versa.
\end{definition}
Note that coupled agents must be in the same disjoint agent group $Gr$. Importantly, disjoint agent groups do not need to solely consistent of coupled agents, i.e. it is okay for extra non-coupled agents (e.g. agents independent of all others) to be in disjoint groups.

From an abstract perspective, we can compute this by iterating through each agent $R_i$ which is not on its optimal single-agent path and storing the ids of the other agents which prevented it from picking a better path (there must exist at least one other agent $R_j$ otherwise $R_i$ could have gone on its optimal path). This builds a dependency graph where agents that share an edge $R_i \rightarrow R_j$ denote $R_i$ blocked by $R_j$.
We can then find all the disjoint connect components in this dependency graph (e.g. via a DFS) where each disjoint connected component depicts a group of agents $Gr^i$ that are blocking each other and are independent from other groups. 
Instead of updating $h(\comp) \gets U(\comp,\comp^W)$, we do $h(\comp_{Gr^i}) \gets U(\comp_{Gr^i},\comp_{Gr^i}^W)$ for each group of agents $Gr^i$ at configuration $\comp_{Gr^i}$.

Our objective is, given $\comp,\comp^W$, to detect these groups of agents and apply heuristic penalties to just the group of agents rather than the entire configuration.
One crucial observation is that it can be non-trivial to determine which agents are coupled once the AG returns the next state. Agents could be next to each other but not block each other, or on the flip side be non-adjacent but coupled. However, instead of reasoning about coupled agents after the AG, we can leverage the AG itself as it must have internally reasoned about agent interactions to return a valid next action. All modern MAPF heuristic search planners reason about agent interactions internally, e.g., M* explicitly couples agents that intersect \cite{mstart_2011}, PIBT's priority inheritance reasons about colliding agents, and CBS resolves collisions between intersecting agents. Thus, we require that the AG additionally returns groups of interacting agents. We highlight that this can be generally done with bookkeeping and without much added compute to existing MAPF solvers. 

Given a configuration $\comp$ and a set of HPs for groups of agents at various $\comp_{Gr}$, we then compute $h(\comp) = h_{BD}(\comp) + \sum_i h_p(\comp_{Gr_i})$ for a disjoint set of groups $\{\comp_{Gr_i}\}$ whose locations match the configuration. 

We lastly note that Disjoint Agents Groups is related to Independence Detection in Operator Decomposition \cite{standley2010operater_decomposition} which dynamically constructs independent subproblems by checking if agents block each other.

\subsection{Overall \framework{} Framework} \label{sec:overall-framework}
Altogether, our framework requires an Action Generator AG with the following two properties:
(1) finds $\argmin_{\comp^W} c(\comp, \comp^W) + h(\comp^W)$, and (2) computes disjoint agent groups for $\comp \rightarrow \comp^W$.

Given such an action generator, we are guaranteed to eventually reach the goal if a solution exists via Algorithm \ref{alg:windowed-framework}. We have a set of Heuristic Penalties (line \ref{alg:w-mapf:initial-hps}) which we pass into the AG with the current configuration (line \ref{alg:w-mapf:ag}). The AG returns a partial path corresponding to configuration $\comp^W$ as well as a list of disjoint agent groups. For each group $Gr$, we compute the update equation with respect to the group's configuration $\comp_{Gr}$ and if it is greater than 0, we add the group configuration and penalty into our library of heuristic penalties. Then, we move agents and repeat.

\begin{theorem} \label{thm:w-mapf-complete}
    Given a finite bidirectional graph and: (1) an initial Backward Dijkstra heuristic, (2) our AG picks $\argmin_{\comp^W} c(\comp, \comp^W) + h(\comp^W)$ and identifies disjoint agents groups, then \framework{} with its update equation (Eq. \ref{eq:update}) applied on group configurations is complete, i.e. all agents will eventually reach their goals if a solution exists. 
\end{theorem}

Section \ref{sec:winc-mapf-proof} contains a formal proof. Briefly, the key is showing how $h(\comp)$ remains admissible even with the grouped HPs, which we do by reasoning about the interaction of agents inside of and between disjoint agent groups. 
    
\begin{algorithm}[t]
\caption{Windowed Complete MAPF Framework}
\label{alg:windowed-framework}
\begin{algorithmic}[1] 
\Procedure{Windowed Complete MAPF}{$\comp^{cur}$}
\State $\mathcal{H} \gets \emptyset$ \Comment{Heuristic Penalties} \label{alg:w-mapf:initial-hps}
\While{$\comp^{cur} \neq$ Goal}
    \State $\comp^W$, ListOfGroups = AG($\comp^{cur}$, $\mathcal{H}$) \label{alg:w-mapf:ag}
    \For{$Gr$ $\in$ ListOfGroups} \Comment{For each group}
    \State $h_{n} \gets U(\comp^{cur}_{Gr}, \comp_{Gr}^W)$ \Comment{Equation (1)}
        \State $penalty \gets h_{n} - h_{BD}(\comp_{Gr}^W)$
        \If{$penalty > 0$}
            \State $\mathcal{H}$.insert($\comp^{cur}_{Gr}$, $penalty$)
        \EndIf
    \EndFor
    \State $\comp^{cur} \gets C^{0<i\leq W}$  \Comment{Move at least one step}
\EndWhile
\EndProcedure
\end{algorithmic}
\end{algorithm}

The main assumptions we have currently are that we have a centralized AG and we have a perfect backward Dijkstra heuristic for each agent. The first assumption allows the AG to reason about heuristic penalties between coupled agents. The second assumption is not strictly required but simplifies the problem as we do not need to do single-agent heuristic updates. Both of these assumptions are common in current MAPF literature, and prior work has shown that for certain state-of-the-art methods like LaCAM, a perfect backward Dijkstra heuristic is required \cite{veerapaneni2024improving_mapf_policies_with_search}. Both of these assumptions can be relaxed in future work.

Section \ref{sec:ss-cbs} describes Single-Step CBS, an action generator that satisfies the two properties required by \framework{}.


\begin{figure*}[t!]
    \centering
    \includegraphics[width=0.95\textwidth]{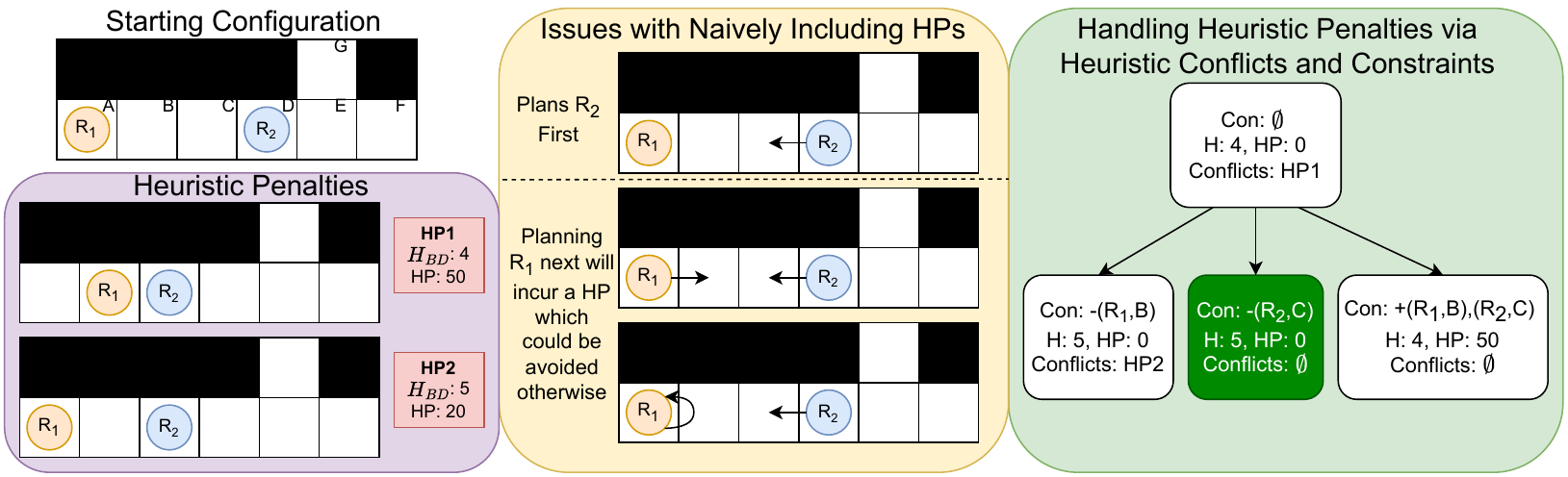}
    \caption{
    We depict an illustrative situation where SS-CBS needs to determine the best next configuration given heuristic penalties (HP, left). Section \ref{sec:hp-constraints} describes this figure and how naively incorporating HPs in CBS results in incorrect solutions (middle). Our innovation is to introduce ``heuristic conflicts" which allows SS-CBS to find the optimal solution (right).
    }
    \label{fig:handling-hps}
    \vspace*{-1em}
\end{figure*}

\section{Single-Step CBS} \label{sec:ss-cbs}
We want to design a windowed solver that incorporates heuristic penalties and optimally solves $W=1$, i.e. finds the next best step that minimizes $c(\comp, \comp') + h(\comp')$. We focus on solving $W=1$ as this is still a challenging problem. In MAPF with $N$ agents and an individual action space of size 5, naively computing the optimal action requires generating all $5^N$ possible neighboring configurations as heuristic penalties could arbitrarily be placed! Thus we employ CBS\footnote{Readers unfamiliar with CBS should read \citet{sharon2015cbs}}, an optimal full horizon planner, to intelligently find the optimal single-step configuration.

Since we want a windowed solver with $W=1$, Single-Step CBS (SS-CBS) only considers conflicts within the first timestep. 
However, regular CBS does not operate in the configuration space of MAPF problems and instead exploits the structure of MAPF to iteratively plan agents individually. One key assumption in CBS is that it can minimize the joint $c(\comp, \comp') + h(\comp')$ by minimizing individual $c(s,s') + h(s')$ subject to constraints. Thus, incorporating HPs which work on the configuration of disjoint agent groups breaks this assumption and requires careful reasoning.
Our main innovation lies in modifying SS-CBS to return the optimal solution given heuristic penalty updates. A minor additional modification is returning disjoint agent groups.

\subsection{Handling Heuristic Penalties with Constraints} \label{sec:hp-constraints}
Incorporating updates to $h$ via heuristic penalties is non-trivial in CBS. A naive way to incorporate a HP in CBS is to plan CBS normally and just add the HP cost to Constraint Tree (CT) nodes whose configurations match the penalty. However, we show that this fails to find an optimal solution. A second naive way is to incorporate the penalty in the low-level search which similarly fails. The core conceptual issue with naively incorporating HPs is that the high-level/low-level is not fully aware of the HPs until \textit{after} it has planned paths, so it is unable to avoid them beforehand. 

Figure \ref{fig:handling-hps} depicts Single-Step CBS in a scenario where two agents want to swap their locations from their starting configurations (i.e., $R_1,R_2$ want to reach $D,A$ respectively, top left). The left purple box shows two HPs with their corresponding configurations and penalty values (50 and 20 respectively). We created the two HPs for this example; in the real system, HPs would be created from previous iterations of execution and applying the subgroup logic and Equation \ref{eq:update} described in Section \ref{sec:coupled-agents}. Given the starting configuration and the two HPs, SS-CBS needs to find the optimal next configuration that minimizes $c(\comp,\comp') + h(\comp')$. Since the cost of all actions is 1, $c(\comp, \comp') = 2$ regardless of the chosen $\comp'$ configuration. Thus, we focus on minimizing $h(\comp')$. In our example, the optimal next configuration is $(R_1,B), (R_2,D)$ which has a heuristic of 5.

\textbf{Incorrect: Incorporating HPs in High Level}
The most obvious way to incorporate HPs is to add them to the CT node's heuristic if the CT node's configuration matches the penalty. This fails as the penalty is applied \textit{after} the low-level planning occurs, so the low-level planner does not avoid HPs in the beginning.

We take a look at generating the root CT node in our example (middle yellow box). If the root node first plans $R_2$, then $R_2$ moves to $C$ as this reduces its heuristic (middle box, top row). We do not incur a HP as we do not know the configuration of $R_1$ yet. When we plan $R_1$, the low-level search has $R_1$ minimize its single agent heuristic and move to $B$ which results in the root CT node with $(R_1,B),(R_2,C)$ (middle box, middle row). This then incurs HP1's penalty of 50. Note that there are no agent collisions and thus CBS will return the single-step solution with a net-heuristic of $4+50=54$. This is substantially worse than the true single-step solution.

\textbf{Incorrect: Incorporating HPs in Low Level}
On the flip side, we could attempt to incorporate the heuristic penalty in the low-level search. 
When replanning an agent, we know the location of all other agents, so the low-level search can check if certain configurations would incur a heuristic penalty. However, this fails as the first agents that plan in the root CT node do not know the locations of other agents that haven't been planned yet. As a result, they plan greedily, potentially forcing later agents into suboptimal situations.

Like before, we can plan $(R_2,C)$ which does not incur any penalty as we do not have $R_1$'s location. When planning for $R_1$, we know $(R_2,C)$ and the search will penalize $(R_1,B)$ by HP1 and instead picks $(R_1,A)$ which is only penalized by HP2 (middle box, bottom row). This results in a net heuristic of $5+20=25$ which is again not optimal.

\textbf{Solution: Introducing ``Heuristic Conflicts"}
Our idea is thus \textit{not} to incorporate the heuristic penalty immediately. Instead, when CBS encounters a configuration that would incur a penalty, it marks the CT node with a ``heuristic conflict" without adding the penalty into the CT node's heuristic value yet. Formally, a heuristic conflict occurs when agents' locations match the configuration of a HP. Resolving the heuristic conflict requires applying regular (negative) vertex constraints on each agent in the heuristic conflict which forces them to avoid the configuration (and thus penalty) as well as one CT node with positive vertex constraints which requires the agents to be at the penalty configuration and only then incurring the HP\footnote{``Negative" vertex constraints avoid vertices while ``positive" vertex constraints force agents to certain vertices.}.

In our example, the agents plan independently like usual and the root node has no vertex or edge conflicts. However, we detect that HP2 could apply and create the corresponding heuristic conflict (right box, top CT node). We then generate three child nodes, the first two CT nodes with negative vertex constraints and the last one with multiple positive vertex constraints. We see how this results in the optimal configuration being found (highlighted in green). Thus given an HP with $K$ agents, our heuristic constraint will generate $K$ children with a single additional negative vertex constraint and one child with $K$ additional positive vertex constraints.

\subsection{Detecting Disjoint Agent Groups}
SS-CBS should also return disjoint agent groups where each disjoint group contains coupled agents. Our key observation is that coupled agents must have conflict(s) between them that got resolved. 
Similarly, independent agents will not conflict with each other. 

One small subtly; coupled agents must have conflicts, but agents with conflicts may not be coupled, e.g. an agent could tie-break poorly and have an ``unnecessary" conflict. From our definition of disjoint agent groups, having extra agents in groups is acceptable.

There is the possibility of indirect interactions. In particular, $R_1$ could conflict with $R_2$, causing $R_2$ to replan which then conflicts with $R_3$. In this case, $R_1$'s action of $R_2$ directly caused an interaction with $R_3$, so $R_1$ and $R_3$ are indirectly coupled. 
Thus we can determine disjoint groups of coupled agents by first generating non-disjoint groups of agents for each resolved conflict and then merging groups with shared agents. In our example, we start with $(R_1,R_2)$ and $(R_2,R_3)$ and will end with $(R_1,R_2,R_3)$ after merging.

We note that we only care about resolved conflicts that occur on CT nodes on the solution branch that led to the outputted configuration. Thus when planning, once a goal CT node is reached, we backtrack and store all resolved vertex, edge, and heuristic conflicts as individual groups, and then merge groups together to get our disjoint agent groups. 

\subsection{Subtleties}
We have two additional subtleties. First, there is an important implementation nuance for determining which HP to apply to a configuration if multiple HPs were applicable. Second, we introduced a tiebreaking mechanism that, when faced with two solutions of equal cost, prioritizes reducing the heuristic of some agents over others. Both are discussed in detail in the appendix.

\begin{figure*}[t!]
    \centering
    \includegraphics[width=0.9\textwidth]{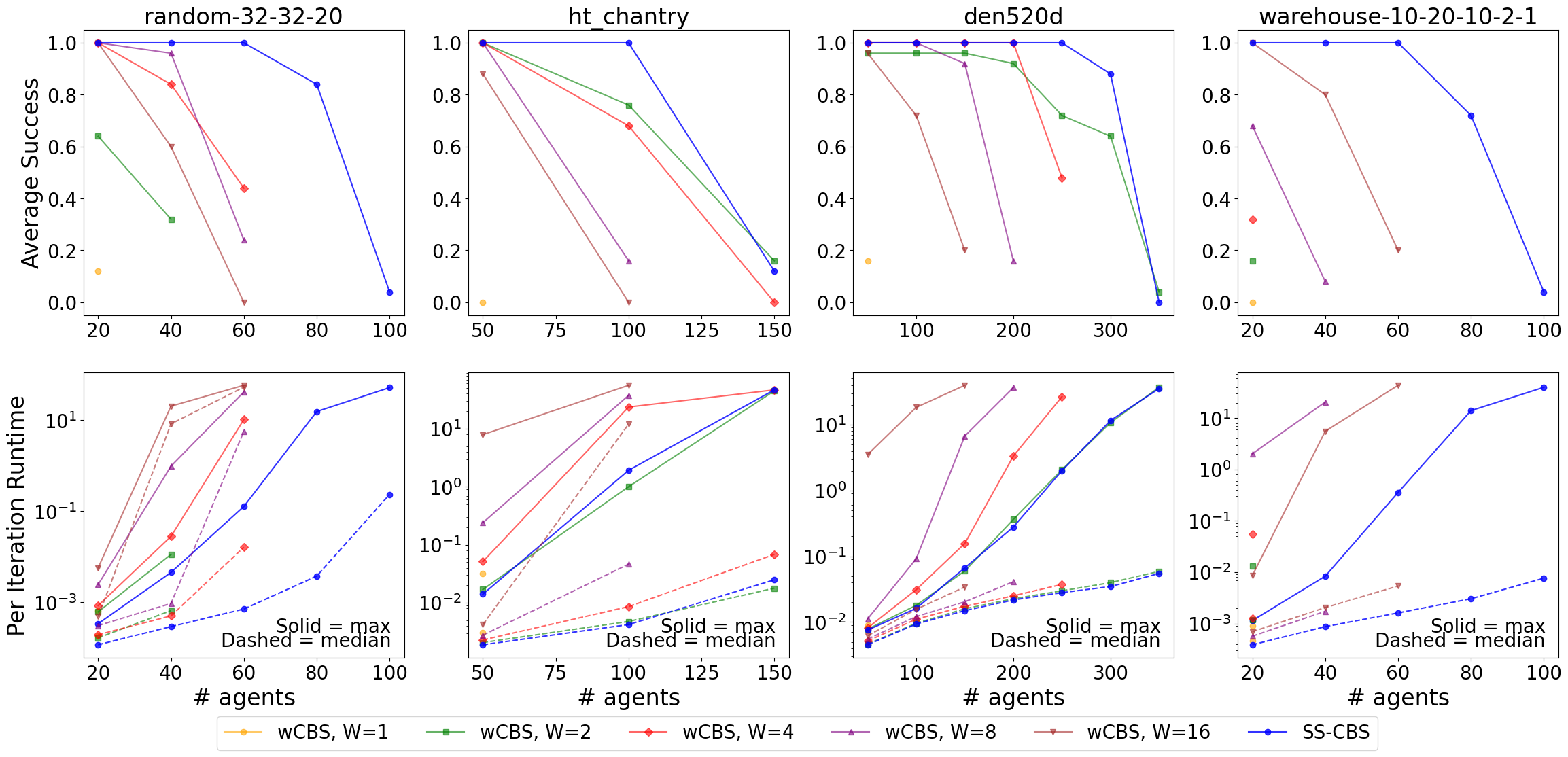}
    \caption{We compare our method SS-CBS (blue), which has $W=1$, against windowed CBS (wCBS) with different window sizes $W=\{1, 2, 4, 8, 16\}$. SS-CBS is theoretically complete (wCBS is not) and also outperforms wCBS empirically.
    }
    \label{fig:main-results-plot}
    \vspace*{-1em}
\end{figure*}

\section{Experiments}
Our experiments demonstrate the empirical performance of our theoretically complete SS-CBS algorithm. We first evaluate SS-CBS on standard benchmark maps \cite{stern2019mapfbenchmark} and observe that SS-CBS is indeed able to outperform the windowed baselines. We then evaluate SS-CBS on high-congestion small maps and showcase SS-CBS's superiority in this regime. Our appendix contains additional analysis and results. We highlight that there do not exist any complete windowed baselines. We thus compare against CBS with window $W = \{1, 2, 4, 8, 16\}$ (wCBS). All algorithms were implemented in C++ and run on a PC with a 2.30 GHz Intel i7-11800K CPU.

\subsection{Benchmark Scenarios}
Figure \ref{fig:main-results-plot} shows the results of SS-CBS compared to wCBS without heuristic penalties. We evaluate on 4 benchmark maps (each column: random-32-32-20, ht\_chantry, den520d, warehouse-10-20-10-2-1) with a 1 minute timeout across 25 scenarios. We additionally end wCBS in failure when it repeats previously visited configurations 100 times (i.e. deadlock or livelock).

The first row shows how SS-CBS (blue) has an almost strictly better performance than the wCBS baselines. We first highlight that the performance of wCBS depends significantly on the map, with no single window value dominating. wCBS with $W=1,2$ sometimes perform poorly due to their small window size which leads to deadlock/livelock, while $W=8, 16$ sometimes suffer from large runtimes. 
SS-CBS is able to consistently perform better than the corresponding wCBS method even though SS-CBS plans only a single step. SS-CBS's performance on warehouse-10-20-2-1 shows how it performs well in scenarios that require long-term planning (as shown by wCBS requiring a larger window size to have non-trivial success rate).

The second row shows the per-iteration runtime of each method, with the median runtime in solid and the maximum (longest) iteration runtime in dashed. \textbf{Limitation:} The difference between the median and maximum highlights how congested iterations can take a significant amount of time (e.g. 10s of seconds for one optimal step) compared to typical iterations. All the failure instances for SS-CBS occurred due to the large runtime of those bottleneck iterations. We observe how SS-CBS's runtime is usually in line with wCBS $W=1,2,4$ and substantially smaller than $W=8,16$.

The solution cost of SS-CBS is roughly 2-8\% higher than wCBS due to SS-CBS's myopic planning window. We observed that the number of HPs created by SS-CBS grew roughly linearly along with the number of agents, but the number of HPs encountered in the CT grew exponentially.



\subsection{Tough Congested Scenarios}
The previous section shows how SS-CBS can outperform wCBS in standard benchmark scenarios. We were additionally interested in SS-CBS's performance in extremely congested scenarios and thus evaluated it on small maps with high congestion from \citet{okumura2022lacam}.

Table \ref{tab:sneak-peak} shows the average success from 20 seeds with a timeout of 1 minute, with the number $N$ in parenthesis denoting using the first $N$ agents in a scenario. We see how windowed CBS with $W=1,2,4,8,16$ fails completely. 
Additionally, CBS with all CBS improvements, including bypass, symmetry reasoning, prioritized conflicts, and \textit{full horizon planning} (CBS+) struggles.

Table \ref{tab:small-results-appendix} shows the average runtime in \textit{milliseconds}. On solved instances, SS-CBS is much faster compared to EECBS with a suboptimality $w_{so}$ of 1.5 and CBS+. \textbf{Limitation:} However, the drawback of using SS-CBS here is that its solution cost is horrible. It is about 7500 for Tunnel (4), 11000 for Loopchain (7), and 450 for Connector (6).

\begin{table}[t!]
\setlength{\tabcolsep}{3pt}
\resizebox{0.99\linewidth}{!}{
\begin{tabular}{c||rr|rr|rr}
 & \multicolumn{2}{c}{SS-CBS} & \multicolumn{2}{c}{CBS+} & \multicolumn{2}{c}{EECBS $w_{so}=1.5$} \\
Instance & Success & Time & Success & Time & Success & Time \\ \hline
Tunnel (3) & 1 & 85 & 0.9 & 43,423 & 0.80 & 3,768 \\
Tunnel (4) & 1 & 2,713 & 0 & - & 0 & - \\ \hline
Loopchain (6) & 1 & 15,393 & 0 & - & 0 & - \\
Loopchain (7) & 0.95 & 18,268 & 0 & - & 0 & - \\ \hline
Connector (5) & 1 & 15 & 0 & - & 0.95 & 146 \\
Connector (6) & 1 & 43 & 0 & - & 0.90 & 1,116
\end{tabular}
}
\caption{SS-CBS with \textit{one-step planning} can solve problems that even CBS and EECBS \textit{with full horizon planning} and all optimizations cannot. 
This highlights SS-CBS's and the broader \framework{} framework's ability to effectively use grouped heuristic penalties to resolve congestion. 
Time (ms) is averaged over successful instances.}
\label{tab:small-results-appendix}
\vspace{-1em}
\end{table}

SS-CBS's performance demonstrates two key points. First, its high success rate on these challenging, congested maps shows how heuristic penalties can effectively guide SS-CBS to bypass complex congestion. Second, full horizon CBS+ struggles due to the high number of conflicts in these scenarios. SS-CBS's success highlights how iterative single-step planning with heuristic updates can be an effective approach for resolving difficult congestion issues.
\textit{All existing windowed methods produce congestion} due to their myopic planning. Thus, SS-CBS improved success rate with windowed planning in severe congestion is remarkable.

\section{Conclusion and Future Work} \label{sec:conclusion}
Existing MAPF works have focused on designing methods to solve full-horizon planning. When faced with shorter deadlines, all current methods take the full horizon MAPF methods and simply reduce the planning horizon to a smaller window. This has been shown to cause deadlock/livelock/insurmountable congestion due to the limited planning horizon. We introduce \framework, the first framework that enables theoretical completeness with windowed MAPF solvers. In particular, we show that using windowed ``Action Generator" that incorporates heuristic penalties, identifies disjoint agent groups, and optimally minimizes $c(\comp, \comp^W) + h(\comp^W)$ is complete. Following this framework, we designed SS-CBS which uniquely introduces ``heuristic conflicts" to successfully incorporate heuristic penalties and return the optimal next step action.
We experimentally validate how our theoretically complete method actually translates to real performance benefits with SS-CBS consistently outperforming windowed CBS across a variety of windows and maps. 

We are very excited about future work that can relax the limitations of our framework.
First, the most useful extension is to incorporate and prove completeness for bounded or even abitrarily suboptimal AGs within our framework (we are limited to optimal AGs). This can allow more solvers such as PIBT, MAPF-LNS2 \cite{li2022mapf-lns2}, W-EECBS \cite{effectiveCBS}, or even using a learnt neural network policy. 
Second, one obvious extension is to generalize SS-CBS to work on longer horizons. 
This also opens the door to using more sophisticated real-time heuristic update methods.
Third, future work could try to relax the perfect backward Dijkstra single-agent heuristic assumption and learn individual heuristics online.

Planning partial paths rather than full paths is a significant target that researchers need to achieve for their methods to be used in real systems.
We believe the \framework{} framework and SS-CBS are a significant step towards bridging this gap and enabling effective windowed MAPF solvers.

\section*{Acknowledgments}
The research was supported by the National Science Foundation under grant \#2328671, by the National Science Foundation Graduate Research Fellowship Program under grant \#DGE2140739, and a gift from Amazon. The views and conclusions in this document are those of the authors and should not be interpreted as representing the official policies, either expressed or implied, of the sponsoring organizations, agencies, or the U.S. government.

\bibliography{ref}

\begin{thebibliography}{25}
\providecommand{\natexlab}[1]{#1}

\bibitem[{Barer et~al.(2014)Barer, Sharon, Stern, and Felner}]{barer2014suboptimalecbs}
Barer, M.; Sharon, G.; Stern, R.; and Felner, A. 2014.
\newblock Suboptimal variants of the conflict-based search algorithm for the multi-agent pathfinding problem.
\newblock In \emph{Seventh Annual Symposium on Combinatorial Search}.

\bibitem[{Cardei and Du(2005)}]{cardei2005disjoint_set_cover}
Cardei, M.; and Du, D. 2005.
\newblock Improving Wireless Sensor Network Lifetime through Power Aware Organization.
\newblock \emph{Wirel. Networks}, 11(3): 333--340.

\bibitem[{Chan et~al.(2024)Chan, Chen, Guo, Zhang, Zhang, Harabor, Koenig, Wu, and Yu}]{chan2024_league_robot_runners_competition}
Chan, S.-H.; Chen, Z.; Guo, T.; Zhang, H.; Zhang, Y.; Harabor, D.; Koenig, S.; Wu, C.; and Yu, J. 2024.
\newblock The League of Robot Runners Competition: Goals, Designs, and Implementation.
\newblock In \emph{ICAPS 2024 System's Demonstration track}.

\bibitem[{Erdmann and Lozano-Perez(1987)}]{erdmann1987multiple}
Erdmann, M.; and Lozano-Perez, T. 1987.
\newblock On multiple moving objects.
\newblock \emph{Algorithmica}, 2(1): 477--521.

\bibitem[{Jiang et~al.(2024)Jiang, Zhang, Veerapaneni, and Li}]{jiang2024scaling_mapf_competition}
Jiang, H.; Zhang, Y.; Veerapaneni, R.; and Li, J. 2024.
\newblock Scaling Lifelong Multi-Agent Path Finding to More Realistic Settings: Research Challenges and Opportunities.
\newblock In \emph{Proceedings of the International Symposium on Combinatorial Search}, volume~17, 234--242.

\bibitem[{Koenig and Likhachev(2006)}]{koenig2006_rtaa}
Koenig, S.; and Likhachev, M. 2006.
\newblock Real-time adaptive A*.
\newblock In Nakashima, H.; Wellman, M.~P.; Weiss, G.; and Stone, P., eds., \emph{5th International Joint Conference on Autonomous Agents and Multiagent Systems {(AAMAS} 2006), Hakodate, Japan, May 8-12, 2006}, 281--288. {ACM}.

\bibitem[{Koenig and Sun(2009)}]{koenig2009_lss_lrta}
Koenig, S.; and Sun, X. 2009.
\newblock Comparing real-time and incremental heuristic search for real-time situated agents.
\newblock \emph{Autonomous Agents and Multi-Agent Systems}, 18: 313--341.

\bibitem[{Korf(1990)}]{korf1990_lrta}
Korf, R.~E. 1990.
\newblock Real-time heuristic search.
\newblock \emph{Artificial Intelligence}, 42(2): 189--211.

\bibitem[{Li et~al.(2021{\natexlab{a}})Li, Chen, Harabor, Stuckey, and Koenig}]{li2021mapf-lns}
Li, J.; Chen, Z.; Harabor, D.; Stuckey, P.~J.; and Koenig, S. 2021{\natexlab{a}}.
\newblock Anytime Multi-Agent Path Finding via Large Neighborhood Search.
\newblock In Dignum, F.; Lomuscio, A.; Endriss, U.; and Now{\'{e}}, A., eds., \emph{{AAMAS} '21: 20th International Conference on Autonomous Agents and Multiagent Systems, Virtual Event, United Kingdom, May 3-7, 2021}, 1581--1583. {ACM}.

\bibitem[{Li et~al.(2022)Li, Chen, Harabor, Stuckey, and Koenig}]{li2022mapf-lns2}
Li, J.; Chen, Z.; Harabor, D.; Stuckey, P.~J.; and Koenig, S. 2022.
\newblock MAPF-LNS2: Fast Repairing for Multi-Agent Path Finding via Large Neighborhood Search.
\newblock \emph{Proceedings of the AAAI Conference on Artificial Intelligence}, 36(9): 10256--10265.

\bibitem[{Li et~al.(2021{\natexlab{b}})Li, Harabor, Stuckey, and Koenig}]{srli2021}
Li, J.; Harabor, D.; Stuckey, P.~J.; and Koenig, S. 2021{\natexlab{b}}.
\newblock Pairwise Symmetry Reasoning for Multi-Agent Path Finding Search.
\newblock \emph{CoRR}, abs/2103.07116.

\bibitem[{Li, Ruml, and Koenig(2021)}]{li2021eecbs}
Li, J.; Ruml, W.; and Koenig, S. 2021.
\newblock EECBS: A bounded-suboptimal search for multi-agent path finding.
\newblock In \emph{Proceedings of the AAAI Conference on Artificial Intelligence (AAAI)}, 12353--12362.

\bibitem[{Li et~al.(2020)Li, Tinka, Kiesel, Durham, Kumar, and Koenig}]{rhcrLi2020}
Li, J.; Tinka, A.; Kiesel, S.; Durham, J.~W.; Kumar, T. K.~S.; and Koenig, S. 2020.
\newblock Lifelong Multi-Agent Path Finding in Large-Scale Warehouses.
\newblock In \emph{Proceedings of the 19th International Conference on Autonomous Agents and MultiAgent Systems}, AAMAS '20, 1898–1900. Richland, SC: International Foundation for Autonomous Agents and Multiagent Systems.
\newblock ISBN 9781450375184.

\bibitem[{Okumura(2022)}]{okumura2022lacam}
Okumura, K. 2022.
\newblock LaCAM: Search-Based Algorithm for Quick Multi-Agent Pathfinding.
\newblock arXiv:2211.13432.

\bibitem[{Okumura et~al.(2022)Okumura, Machida, Défago, and Tamura}]{pibt}
Okumura, K.; Machida, M.; Défago, X.; and Tamura, Y. 2022.
\newblock Priority inheritance with backtracking for iterative multi-agent path finding.
\newblock \emph{Artificial Intelligence}, 310: 103752.

\bibitem[{Rivera, Baier, and Hernandez(2013)}]{rivera2013weighted_real_time_search}
Rivera, N.; Baier, J.~A.; and Hernandez, C. 2013.
\newblock Weighted real-time heuristic search.
\newblock In \emph{Proceedings of the 2013 International Conference on Autonomous Agents and Multi-Agent Systems}, AAMAS '13, 579–586. Richland, SC: International Foundation for Autonomous Agents and Multiagent Systems.
\newblock ISBN 9781450319935.

\bibitem[{Sharon et~al.(2015)Sharon, Stern, Felner, and Sturtevant}]{sharon2015cbs}
Sharon, G.; Stern, R.; Felner, A.; and Sturtevant, N.~R. 2015.
\newblock Conflict-based search for optimal multi-agent pathfinding.
\newblock \emph{Artificial Intelligence}, 219: 40--66.

\bibitem[{Sigurdson et~al.(2018)Sigurdson, Bulitko, Yeoh, Hernández, and Koenig}]{mapf-real-time-bmaa-2018}
Sigurdson, D.; Bulitko, V.; Yeoh, W.; Hernández, C.; and Koenig, S. 2018.
\newblock Multi-Agent Pathfinding with Real-Time Heuristic Search.
\newblock In \emph{2018 IEEE Conference on Computational Intelligence and Games (CIG)}, 1--8.

\bibitem[{Silver(2005)}]{cooperativeSilver2005}
Silver, D. 2005.
\newblock Cooperative Pathfinding.
\newblock In \emph{Proceedings of the First AAAI Conference on Artificial Intelligence and Interactive Digital Entertainment}, AIIDE'05, 117–122. AAAI Press.

\bibitem[{Standley(2010)}]{standley2010operater_decomposition}
Standley, T.~S. 2010.
\newblock Finding Optimal Solutions to Cooperative Pathfinding Problems.
\newblock In Fox, M.; and Poole, D., eds., \emph{Proceedings of the Twenty-Fourth {AAAI} Conference on Artificial Intelligence, {AAAI} 2010, Atlanta, Georgia, USA, July 11-15, 2010}, 173--178. {AAAI} Press.

\bibitem[{Stern et~al.(2019)Stern, Sturtevant, Felner, Koenig, Ma, Walker, Li, Atzmon, Cohen, Kumar, Boyarski, and Bartak}]{stern2019mapfbenchmark}
Stern, R.; Sturtevant, N.~R.; Felner, A.; Koenig, S.; Ma, H.; Walker, T.~T.; Li, J.; Atzmon, D.; Cohen, L.; Kumar, T. K.~S.; Boyarski, E.; and Bartak, R. 2019.
\newblock Multi-Agent Pathfinding: Definitions, Variants, and Benchmarks.
\newblock \emph{Symposium on Combinatorial Search (SoCS)}, 151--158.

\bibitem[{Veerapaneni, Kusnur, and Likhachev(2023)}]{effectiveCBS}
Veerapaneni, R.; Kusnur, T.; and Likhachev, M. 2023.
\newblock Effective Integration of Weighted Cost-to-Go and Conflict Heuristic within Suboptimal {CBS}.
\newblock In \emph{Thirty-Seventh {AAAI} Conference on Artificial Intelligence, {AAAI} 2023, Washington, DC, USA, February 7-14, 2023}, 11691--11698. {AAAI} Press.

\bibitem[{Veerapaneni et~al.(2024)Veerapaneni, Wang, Ren, Jakobsson, Li, and Likhachev}]{veerapaneni2024improving_mapf_policies_with_search}
Veerapaneni, R.; Wang, Q.; Ren, K.; Jakobsson, A.; Li, J.; and Likhachev, M. 2024.
\newblock Improving Learnt Local MAPF Policies with Heuristic Search.
\newblock \emph{International Conference on Automated Planning and Scheduling}, 34(1): 597--606.

\bibitem[{Wagner and Choset(2011)}]{mstart_2011}
Wagner, G.; and Choset, H. 2011.
\newblock M*: A complete multirobot path planning algorithm with performance bounds.
\newblock In \emph{2011 IEEE/RSJ International Conference on Intelligent Robots and Systems}, 3260--3267.

\bibitem[{Zhang et~al.(2024)Zhang, Chen, Harabor, Bodic, and Stuckey}]{zhang2024pie}
Zhang, Y.; Chen, Z.; Harabor, D.; Bodic, P.~L.; and Stuckey, P.~J. 2024.
\newblock Planning and Execution in Multi-Agent Path Finding: Models and Algorithms.
\newblock \emph{Proceedings of the International Conference on Automated Planning and Scheduling}, 34(1): 707--715.

\end{thebibliography}

\clearpage

\appendix

\setcounter{figure}{0}
\renewcommand{\thefigure}{A\arabic{figure}}
\setcounter{table}{0}
\renewcommand{\thetable}{A\arabic{table}}

\section{Quick Summary} \label{sec:apdx-summary}
\subsubsection*{Recommended background readings} Readers new to MAPF or CBS are recommended to read CBS \cite{sharon2015cbs}. Readers unfamiliar with windowed MAPF solvers should read Rolling-Horizon Collision Resolution \cite{rhcrLi2020}. Readers new to Real-Time Heuristic Search can take a look at LRTA* \cite{korf1990_lrta}.

\subsubsection{Motivation in respect to prior work: }
The majority of MAPF methods which find entire collision-free paths to goal can take a long time (e.g. $>10$ seconds). Real-world practitioners cannot wait this long. Thus, to reduce planning time, existing works only reason about collisions within a fixed time window/horizon, where the window $W$ is much smaller than the length of the entire solution path. 

A key issue with these windowed approaches is that their myopic planning results in deadlock or livelock if their window is too small. 
Table \ref{tab:sneak-peak} shows examples where windowed MAPF solvers fail in congestion which requires longer horizon planning.
More broadly, all existing windowed MAPF solvers regardless of window size lack theoretical completeness, and several windowed works have explicitly cited deadlock as a key issue in their experiments \cite{rhcrLi2020,pibt,jiang2024scaling_mapf_competition}. 


\subsubsection*{Intended Takeaways} \hfill \par
1. \underline{Win}dowed \underline{C}omplete MAPF (\framework) Framework: We develop the first general framework that enables creating windowed MAPF solvers that have completeness guarantees. Our first insight is that we can leverage the single-agent Real-Time Heuristic Search perspective which uses limited horizon/windowed planning but maintains completeness by updating the heuristics of visited states. This, however, does not lead to a practical algorithm as it requires exploring the entire joint configurations to get out of deadlocks/livelocks. Our second insight is to leverage the semi-independence of agents in MAPF and only compute the heuristic updates in respect to groups of coupled agents rather than the entire joint configuration space.

Formally, we define a windowed MAPF ``Action Generator" (AG) as a search method that, given a window $W$ and a current configuration $\comp$, finds a $\comp^W$ (and associated actions) that is a valid neighboring configuration within $W$ timesteps. Additionally, the AG needs to determine groups of coupled agent groups as defined in Section \ref{sec:coupled-agents}. We prove in the next section how an optimal windowed AG that computes $\argmin c(\comp, \comp^W) + h(\comp^W)$ is complete under regular MAPF conditions. 

2. Single-Step CBS (SS-CBS): We develop SS-CBS which is an instantiation of an AG that can be used within our \framework{} framework. SS-CBS finds the optimal next step (so window $W=1$) given heuristic updates. We show that naively incorporating heuristic updates (which we redefine as heuristic ``penalties") by adding it into the high-level or low-level CBS search is incorrect. SS-CBS's innovation is to incorporate heuristic penalties by introducing a new ``heuristic conflict" and constraint that defers the addition of the heuristic penalty and enables all agents to replan to avoid a penalty. We additionally show how SS-CBS can easily determine coupled agent groups by merging pairs of conflicting agents.

Experimentally, Figure \ref{fig:main-results-plot} shows how SS-CBS has a higher success rate and agent scalability on standard MAPF benchmark maps compared to windowed CBS across windows $W = \{1, 2, 4, 8, 16\}$. We highlight that SS-CBS plans only one single step (i.e., extremely myopic planning) but is able to handle congestion/deadlock better. Additionally, we investigate the performance of these windowed methods in tough small scenarios (Tables \ref{tab:sneak-peak} and \ref{tab:small-results-appendix}). In these scenarios, we find that windowed CBS methods uniformly fail, and that even CBS with all optimizations struggles. SS-CBS is able to almost perfectly solve these instances. All existing windowed MAPF methods struggle in congestion, so SS-CBS's superior performance in these congested scenarios highlights the power of the \framework{} framework.

\subsubsection{Main Limitations and Future Improvements} \hfill \par
1. SS-CBS's main limitation is that although most iterations are fast ($<$0.1 seconds), Figure \ref{fig:main-results-plot} shows how a few iterations can take a significant amount of time (e.g. $>$10 seconds) when congestion increases. Future work should improve the runtime of SS-CBS.

2. SS-CBS requires a perfect single-agent heuristic. Although most MAPF methods assume this, this is not generally required for CBS and future work could relax this.

3. SS-CBS plans a single step. Future work can extend SS-CBS to planning multiple steps.

4. Our \framework{} currently proves completeness with an optimal AG. We conjecture that future work can likely modify and prove that bounded suboptimal AG's can be complete within our framework.

\section{Proving Completeness of \framework{} with Group Heuristic Penalties}

A key idea of \framework{} is to compute heuristic penalties (HPs) on groups of agents rather than the entire configuration. 
We first roughly restate the standard single-agent proof used in the original real-time search LRTA* paper \cite{korf1990_lrta} to provide context on the differences that our \framework{} framework has in respect to proving completeness. We cannot directly use LRTA*'s proof due to our use of agent group HPs. Instead, we can use the fact that this proof shows how (1) proving that an algorithm does not cycle infinitely proves completeness and (2) how heuristic values with cycles with an update of $h'(s^i) \gets c(s^i, s^{i+1}) + h(s^{i+1})$ results in infinitely large heuristic values. We then show how our heuristic values are admissible, which, combined with (1) and (2) with optimal windowed AGs, result in our framework being complete.

\subsection{Standard Single-Agent Real-Time Search Completeness Proof} \label{sec:lrta-proof}
LRTA* \cite{korf1990_lrta} is a single-agent algorithm where at each timestep, the agent picks $\argmin_{s'} c(s,s') + h(s')$ over successor states $s'$ and updates the heuristic $h'(s) \gets c(s, s') + h(s')$. We summarize its proof of completeness.

\begin{theorem}
    In a finite bidirectional graph with positive edge costs and finite heuristic values, in which a goal state is reachable from every state, LRTA* will find a solution. 
\end{theorem}
\begin{proof}
    In a finite bidirectional graph, if LRTA* does not reach the goal, there must exist a finite cycle that LRTA* is stuck in (otherwise it will visit new states and eventually the goal).

    Suppose LRTA* is stuck in a finite cycle of length N consisting of $s_1, s_2, ..., s_N$, and is at $s_1$. When moving to the next state $s_2$, it updates $h'(s_1) \gets c(s_1, s_2) + h(s_2)$. Our objective is to show that when traversing the cycle once, at least one state $s_i$ has their heuristic value increase. If so, then over an infinite amount of cycling, the heuristic values of at least one state $s_i$ in the cycle will become infinitely large. Then when LRTA* is at $s_{i-1}$ and does a 1-step lookahead, it will pick a different state and exit the cycle.

    Suppose LRTA* travels the cycle but for all $s_i$ we have that $h'(s_i) \gets c(s_i, s_{i+1}) + h(s_{i+1}) = h(s_i)$. This implies that $h(s_i) > h(s_{i+1})$ as $c(s_i, s_{i+1}) > 0$. This occurs over all $s_i$ which leads to $h(s_1) > h(s_2) > ... > h(s_N) > h(s_1)$ which is a contradiction. Therefore there must be some state that gets its heuristic value increased, e.g. $\exists s_i$ s.t. $h'(s_i) \gets c(s_i, s_{i+1}) + h(s_{i+1}) > h(s_i)$.

    Thus, LRTA* is guaranteed to eventually leave any finite cycle. Since a finite bidirectional graph has a finite set of cycles without the goal, LRTA* at worst will explore all of these and eventually exit them and reach the goal. 
\end{proof}

\subsubsection{Subtleties}
This proof for completeness assumes that $h'(s_i) \gets c(s_i, s_{i+1}) + h(s_{i+1}) \geq h(s_i)$. This holds when $h$ is consistent but may not hold when $h$ is admissible. Thus for arbitrary $h$ we should update via $h'(s_i) \gets \max(h(s_i), c(s_i, s_{i+1}) + h(s_{i+1}))$ (Eq. \ref{eq:update}) and then reuse the same completeness proof.

This proof also assumes that $s_{i-1}$ has a different state that can be picked, i.e. it has more than one neighbor. If there exists a path in the cycle to the goal, then there must be some $s_{i-1}$ where this holds and we apply this logic for that $s_{i-1}$.

Lastly, this proof was shown for one-step planning but can be directly applied to planning for $W$ steps. It is also independent of how many steps along the $W$ length partial path the agent chooses to move (as long as it moves at least one step along the path). 

\subsection{\framework{} Proof} \label{sec:winc-mapf-proof}
From a high level, the LRTA* proof proves completeness by showing that in a cycle and given the update $h'(s_i) \gets c(s_i, s_{i+1}) + h(s_{i+1})$, the heuristic values in the cycle must increase infinitely large to the extent that LRTA* will pick a different $s_i'$ that avoids the large heuristic value. This requires/assumes that only heuristic values in the cycle increases while other heuristic values do not. 

A key idea of \framework{} is to compute heuristic penalties on the configuration of disjoint agent groups rather than the entire configuration. Note that without this, i.e., if we just compute heuristic penalties on the full configuration, we get completeness by directly applying the single-agent proof. 
Thus, using group HPs complicates the proof of completeness as when computing/creating a heuristic penalty, we increase the heuristic values of configurations not in the cycle.

Instead, we prove completeness in Theorem \ref{thm:w-mapf-complete} by showing how our overall heuristic values are always admissible. If we can show this, then we are guaranteed to not cycle infinitely as this will infinitely increase some heuristic value and contradict our admissible heuristic value guarantee.

\begin{theorem} \label{thm:suppl-joint-admissible}
    Given a finite bidirectional graph and: (1) an initial Backward Dijkstra heuristic, (2) our AG picks $\argmin_{\comp^W} c(\comp, \comp^W) + h(\comp^W)$ and identifies disjoint agents groups, then \framework{} with its update equation (Eq. \ref{eq:update}) applied on group configurations will always have admissible heuristic values for all $\comp$. 
\end{theorem}

We calculate the joint heuristic value by summing up mutually exclusive group HPs as described in Section \ref{sec:apdx-which-hp-apply}. If we can prove that each group HPs values are admissible, then the sums of disjoint grouped heuristic values is also admissible as groups interacting can only increase the true cost-to-go. Thus we prove the following lemma.

\begin{lemma} \label{thm:suppl-grouped-admissible}
    Given a finite bidirectional graph and: (1) an initial Backward Dijkstra heuristic, (2) our AG picks $\argmin_{\comp^W} c(\comp, \comp^W) + h(\comp^W)$ and identifies disjoint agents groups, then \framework{} with its update equation (Eq. \ref{eq:update}) applied on group configurations will always have admissible heuristic values for all group configurations $\comp_{Gr}$. 
\end{lemma}
\begin{proof}
    We prove this via induction. Our inductive hypothesis is that we have $h(\comp_{Gr}) \leq h^*(\comp_{Gr})$ across all group $Gr$ configurations $\comp_{Gr}$ and iterations of running our Windowed MAPF framework. 

    \textbf{Base case:} Our initial group heuristic values, the sum of each agent's backward Dijkstra distance, is admissible as interactions between agents can only increase the solution cost. Thus over all groups of agents, the joint heuristic value is admissible.

    \textbf{Inductive Step:} We assume that at some timestep $T$ at $\comp$, all the heuristic values are admissible. The AG picks a $\comp^W$ that minimizes $c(\comp, \comp^W) + h(\comp^W)$. Then for each group, we update its configuration (via our heuristic penalty) to $h(\comp_{Gr}) \gets U(\comp_{Gr}, \comp_{Gr}^W)$. 

    Suppose that there is a group $GrA$ at $\comp_{GrA}$ whose heuristic value became inadmissible after applying $U$. This can only occur when the chosen $\comp_{GrA}^W$ was not optimal and there must exist a different $\comp_{GrA}^*$ where $U(\comp_{GrA}, \comp_{GrA}^*)$ would be admissible. 
    However, since we are running on optimal AG, we would have replaced $\comp_{GrA}'$ with $\comp_{GrA}^*$. If this did not interact with any other agents, this would reduce our overall cost without changing any other agents' locations/actions and contradict the optimal AG. If it did interact with other agents, then $GrA$ would need to be larger and include these other agents, violating our definition of disjoint agent groups. Thus, neither of these are possible. This means that our heuristic value stays admissible for all disjoint agent groups after updating.    
\end{proof}

\section{SS-CBS Subtleties}
\subsection{Determining Which HPs to Apply} \label{sec:apdx-which-hp-apply}
As discussed earlier, give a CT node we need to detect heuristic conflicts. This can be done by going through all the HPs and checking if their agent locations match with the CT's configuration.

However, a non-trivial issue occurs as it is possible that multiple non-disjoint heuristic conflicts could be detected. For example, a configuration could have a heuristic penalties involving $(R_1,R_2)$ and another involving $(R_2,R_3)$. We cannot apply both as this would double count the heuristic penalty involving $R_2$.
Thus more broadly, we cannot apply multiple non-disjoint heuristic conflicts as if we were to resolve all of them, this will over-count the heuristic penalty.

Given a set of possible HPs, we could try to maximize the combination of these penalties such that chosen penalties do not have overlapping agents. 
This ends up being a maximum weighted disjoint set-cover problem where each heuristic penalties is a set and the set's weight is the penalty. We note that solving just disjoint set-cover is NP-Complete \cite{cardei2005disjoint_set_cover}, so solving this efficiently is hard. 
We thus instead do a greedy approximation and choose to apply the heuristic conflict with the highest penalty.

One important implementation note is that initially we greedily chose the highest HP that was mutually exclusive (i.e. did not share agents) with prior chosen HPs. This is \textit{wrong} and resulted in deadlock as the order in which HPs are chosen affects the solution. In certain deadlock locations, say $(R_1,R_2,R_3)$, SS-CBS would always first encounter an HP with $(R_1,R_2)$ and apply that. Then later in the same CT search, $(R_1,R_2,R_3)$ is encountered but we did not apply an HP as $R_1$ and $R_2$ are already used up. This resulted in SS-CBS repeatedly picking the deadlock location as even though $(R_1,R_2,R_3)$ accumulated large penalties, SS-CBS would only apply the $(R_1,R_2)$ penalty.

Thus, every time we compute heuristic penalties and conflicts, we recompute it from scratch while not violating constraints. This allows us to initially pick $(R_1,R_2)$ and then afterwards ``reassign" $(R_1,R_2,R_3)$ later on.

\subsection{Tie-Breaking} 
Given the 1-step plans, it is likely that there are many SS-CBS solutions with an equally good cost. We found it helpful to tiebreak by introducing agent ``priorities". The intuition is that given a symmetric situation of two agents in a hallway (where all adjacent configurations have the same summed $h$ value), we prefer one agent to ``push" the other agent away rather than oscillate in the middle.

Concretely, if two CT nodes have the same f-value, h-value, g-value (cost), and number of conflicts, we tie-break nodes by lexicographically comparing the agent's h-values. This means that given equally good options, we would prefer a solution that reduces the first lexicographically sorted agent's heuristic more than other options. 
The agent's lexicographical ordering can be thought of as its tiebreaking priority (i.e., the agent that comes first is the highest priority). 
We additionally found assigning random ``priorities" which are updated as in PIBT helped performance.

\subsection{SS-CBS Ablation Study}
A powerful technique for speeding up CBS is Enhanced CBS \cite{barer2014suboptimalecbs} which replaces the low-level and high-level optimal searches with bounded-suboptimal focal searches. Thus, an obvious idea to potentially improve SS-CBS is to introduce these focal searches. As stated in Section \ref{sec:conclusion}, our proof of completeness only applies to optimal AGs. 
For experimental curiosity, we tested out SS-CBS with a suboptimal high-level search (although we do not prove completeness for this suboptimal AG).

Figure \ref{fig:ablation} shows an ablation study of SS-CBS by varying the high-level suboptimality and tiebreaking mechanism. We test suboptimalities of $w_{so}=1, 1.05, 1.1, 1.5$ as well as explore two tie-breaking variants along with the none tie-breaking base version.
The none tie-breaking version will arbitrarily pick any node that shares the same f, h, g, and number of conflicts.
Using $w_{so}>1$ is not theoretically proven or disproven to be complete and is left for future work to study its theoretical properties.

The first tie-breaking method uses PIBT's priority scheme where agents initial priorities are proportional to their distance to the goal; agents further from the goal have higher priority than agents closer to the goal. During execution, the all agent's priority increases by 1 per timestep unless an agent is at its goal, where the priority is set to zero.
We explored another variant of this where we set the initial priorities randomly instead of based on distance, while keeping the same priority update scheme during execution.

In Figure \ref{fig:ablation}, colors denote different high-level suboptimalities while the linestyle denotes different tie-breaks. The ``-n" corresponds to no tie-breaking scheme, ``-d" to an initial Distance based priority, and ``-r" to an initial random based priority. The first row shows how increasing the suboptimality can improve the scalability of SS-CBS. 
The first row also highlights how including the tie-breaking mechanisms has a small but noticeable improvement over no tie-breaks. A close inspection shows that the random initialization tie-breaking helps more than the PIBT initialization tie-break.

The second row shows the number of iterations of execution required to reach the goal (equivalent to makespan). We see two possible patterns given a problem instance where multiple suboptimalities find a solution. In random-32-32-20 or warehouse, we generally see that increasing the suboptimality increases the number of iterations required compared to lower suboptimalities. Upon visual inspection, this seemed to occur when there were instances of deadlock that needed to be resolved via HPs. Higher suboptimalities meant that instead of trying to navigate outside of encountered HPs, SS-CBS several times chose to pick the same location and incur the higher HP cost. 
However, in the less congested maps ht\_chantry and den520d, we do not see as strong a pattern. Thus higher suboptimality seems to incur this drawback mainly in congested scenarios.
We see again that tie-breaking helps reduce the number of iterations required to solve problems.

The last two rows show how many HPs were created and how many HP conflicts were found in CTs across all search iterations. We see the number of HPs found in search increases exponentially as the number of agents increases, although the total number of HPs created does not increase as much. This highlights how SS-CBS spends more effort on resolving congestion via HPs as the number of agents (and therefore congestion) increases.

Figure \ref{fig:suboptimality} compares the scalability of windowed ECBS and SS-CBS as we increase the high-level suboptimality (the low-level suboptimality is kept at 1). Note the x-axis is not linearly scaled. The y-axis is the highest number of agents the method was able to solve at least 50\% of the instances within the 1 minute timeout. ECBS with window $W=1,2$ failed on the lowest tried number of agents so only $W=4,8,16$ are plotted. 

We see that on 3 of the 4 maps, increasing the suboptimality from 1 to 1.05 leads to an improvement but that higher suboptimalities does not for SS-CBS. We visually observed that SS-CBS would be fine waiting at a location with HP when the suboptimality increases (as the suboptimality factor allowed that solution). This shows how more clever ways of incorporating suboptimal search and real-time searches must be researched, similar to Weighted-LRTA* \cite{rivera2013weighted_real_time_search} which showed a non-trivial adjustment to LRTA* to work with sub-optimal searches. 

We also observe that windowed ECBS does not strictly improve as the suboptimality increases. First, the fact that windowed ECBS with $W=1,2$ failed for all suboptimalities shows how these methods fail due to their limited planning horizon. Similarly, a majority of the instances do not improve with suboptimalities higher than 1.05 due to a combination of deadlock and timing out.

\begin{figure*}[t]
    \centering
    \includegraphics[width=0.95\textwidth]{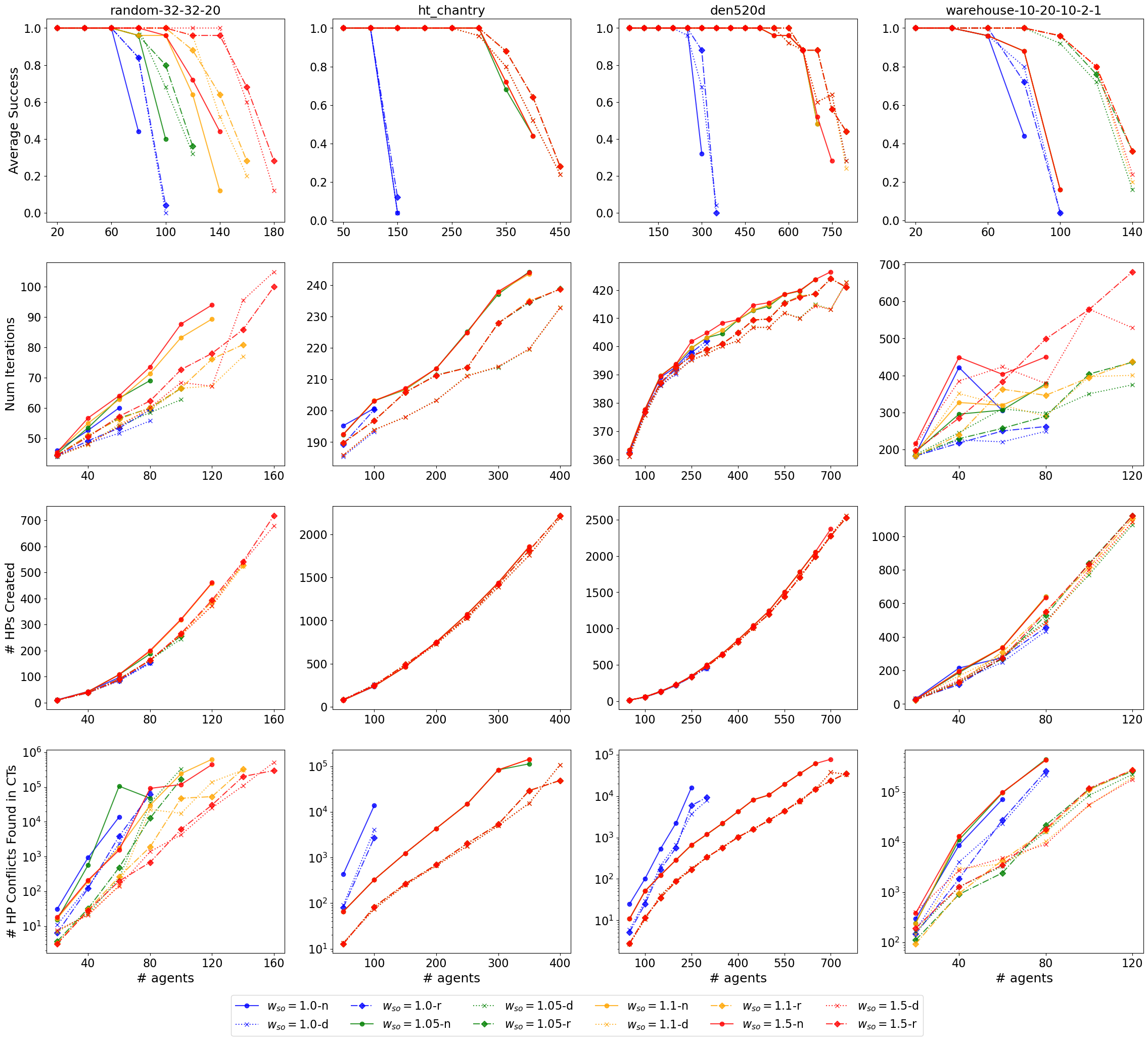}
    \caption{We plot statistics of SS-CBS with different high-level suboptimality $w_{so}$ (colored) and CT tie-breaking (line style). Note that completeness is only proven for an optimal SS-CBS (i.e., $w_{so}=1$ and not $w_{so}>1$).}
    \label{fig:ablation}
    \vspace*{-1em}
\end{figure*}

\begin{figure*}[t]
    \centering
    \includegraphics[width=0.95\textwidth]{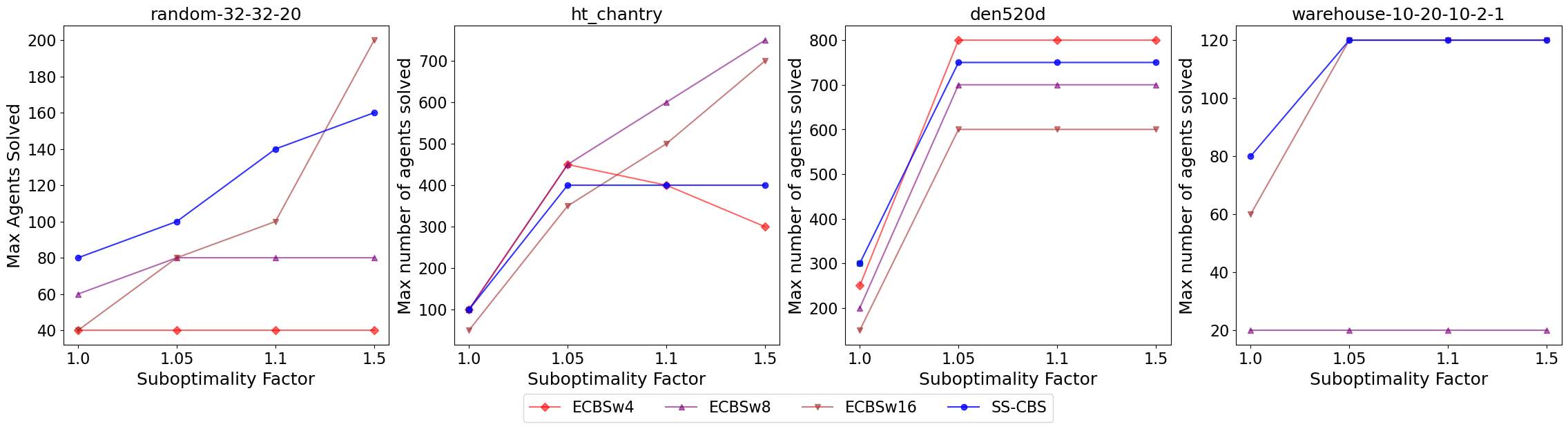}
    \caption{We compare the effects of increasing the high-level suboptimality factor on SS-CBS and windowed ECBS. The y-axis is the highest number of agents the method was able to solve at least 50\% of the instances within the 1 minute timeout. ECBS with window $W=1,2$ failed on the lowest tried number of agents so only $W=4,8,16$ are plotted.}
    \label{fig:suboptimality}
    \vspace*{-1em}
\end{figure*}

\end{document}